\documentclass[11pt]{article}

\usepackage[utf8]{inputenc}
\usepackage{amsmath,amsthm,amssymb,enumerate,framed,mdwlist}
\usepackage{color,colortab}
\usepackage[square,numbers,sort&compress]{natbib}
\usepackage{graphicx}
\usepackage{algorithm}
\usepackage[noend]{algpseudocode}
\usepackage{mathtools}

\usepackage{thmtools}
\usepackage{thm-restate}
\usepackage{hyperref}
\usepackage{cleveref}
\usepackage{tikz}
\usetikzlibrary{calc}
\usetikzlibrary{decorations.pathreplacing}
\usetikzlibrary{decorations.markings}
\usetikzlibrary{arrows}

\newcommand{\R}{\mathbb{R}}

\newcommand{\mb}[1]{\ensuremath{\boldsymbol{#1}}}

\newcommand{\cK}{\mathcal{K}}

\renewcommand{\epsilon}{\varepsilon}

\newtheoremstyle{mythmstyle}
	{\topsep}
	{\topsep}
	{\itshape}
	{}
	{\scshape}
	{.}
	{3pt}
	{}
\theoremstyle{mythmstyle}

\newtheorem{nn}{}[section]
\newtheorem{lemma}[nn]{Lemma}

\newtheorem{prop}[nn]{Proposition}

\newtheorem{claim}[nn]{Claim}

\newtheorem{REMARK}[nn]{Remark}

\newcommand{\comment}[1]{}

\numberwithin{equation}{section}

\allowdisplaybreaks
\usepackage{authblk}
\title{Additional Results and Extensions for the paper ``Probabilistic bounds on the $k$-Traveling Salesman Problem and the Traveling Repairman Problem''}
\author[1]{M. Blanchard}
\author[1]{A. Jacquillat}
\author[1]{P. Jaillet}
\affil[1]{Massachusetts Institute of Technology, Cambridge, MA, USA}
\date{} 
\begin{document}

\maketitle

We study two variants of the classical traveling salesman problem (TSP). Given $n$ points, the TSP seeks a tour of minimal length visiting all $n$ points. In contrast, we focus on
\begin{itemize}
    \item the $k$-TSP which seeks a path of minimal length visiting $k$ out of $n$ points, where $k\leq n$. Formally, if $x_1,\ldots, x_k$ is the service order, the objective to minimize is the path length
    \begin{equation*}
        \sum_{i=1}^{k-1} |x_{i+1}-x_i|.
    \end{equation*}
    \item the traveling repairman problem (TRP) which seeks a tour visiting all $n$ points that minimize the sum of latencies (or waiting time) for each point. Formally, if $x_1,\ldots, x_n$ defines a service order, the latency at point $x_i$ is defined as $l_i = \sum_{j=1}^{i-1}|x_{j+1}-x_j|$ and the objective is to minimize the total latency
    \begin{equation*}
        \sum_{i=1}^n l_i = \sum_{i=1}^{n-1} (n-i) |x_{i+1}-x_i|.
    \end{equation*}
\end{itemize}
We consider a probabilistic setting where $n$ points $X_1,\ldots, X_n$ are sampled independently and identically from some distribution on a compact $\cK\subset \R^2$.

In \cite{blanchard2021probabilistic}, we provided \emph{constant-factor} probabilistic approximations of both problems, i.e., bounds on the expected optimal objective value that hold within a universal constant factor, as well as constant-factor approximation algorithms. Precisely, we show that the optimal length of the $k$-TSP path (non-asymptotically) grows at a rate of $\Theta\left(k/n^{\frac{1}{2}(1+\frac{1}{k-1})}\right)$ and that a constant-factor approximation scheme can be obtained by solving the TSP in a high-concentration zone, leveraging large deviations of local point concentration. Next, we show that the optimal TRP objective follows an asymptotic rate $\Theta(n\sqrt n)$ with a prefactor that depends on the density $f$ of the absolutely-continuous part of the point distribution. This generalizes the classical Beardwood-Halton-Hammersley theorem to the latency-minimization objective in the TRP. The resulting constant-factor approximation scheme visits local regions of the space by decreasing order of probability density $f$. Last, we propose fairness-enhanced versions of the $k$-TSP and the TRP to balance efficiency and fairness.

In this companion report, we provide two additional contributions.
\begin{enumerate}
    \item We extend the $k$-TSP results to the case with general densities. In Section \ref{sec:k-TSP}, we show that the results obtained in \cite{blanchard2021probabilistic} with continuous densities can be extended via smoothing techniques. We also discuss the case of $k=\Omega(n)$, in which case the $k$-TSP path becomes non-local and recovers similar behavior to that of the TRP tour---visiting zones by decreasing order of density until $k$ points are visited.
    \item For the TRP, we propose a utility-based notion of fairness in Section \ref{sec:psi-TRP}. Instead of assuming that the dissatisfaction (or negative utility) of customers is linear in their latency/waiting time, we consider the case where the utility is a convex function $\Psi$ of their latency. A fair solution aims to minimize total dissatisfaction, which we refer to as the $\Psi$-TRP solution. For polynomial functions $\Psi$, we give constant-factor approximations of the optimal $\Psi$-TRP objective, thus extending the TRP bounds to non-linear utility. Further, we show that the approximation scheme for the TRP given in \cite{blanchard2021probabilistic} can be efficiently adapted to obtain constant-factor approximations in the $\Psi$-TRP.
\end{enumerate}

\section{Generalisations of probabilistic bounds for the $k$-TSP}
\label{sec:k-TSP}

In the main paper, we provide probabilistic bounds for the $k$-TSP when points are sampled independently from a distribution with continuous density on a compact. In this section, we present a natural extension of this result to distributions with general densities $f$ on a compact. In particular, the density $f$ is allowed to diverge on a zero-measure set. To this end, we use the notion of Lebesgue derivative $\tilde f$, defined as the local average value of $f$ on centered balls. Intuitively, $\tilde f$ is a smoothed version of the density $f$. For instance, if $f$ is continuous then $\tilde f=f$. Formally, the Lebesgue derivative is defined as follows:
\begin{equation*}
    \tilde f(x) := \lim_{r\to 0}\frac{1}{|B(x,r)|}\int_{B(x,r)} f,\quad\forall x,
\end{equation*}
where $|B(x,r)|$ denotes the volume of a centered ball at $x$ of radius $r$. The Lebesgue differentiation theorem states that this limit exists and that $\tilde f$ and $f$ coincide almost everywhere. By construction, the maximum density of points sampled according to $f$ cannot exceed $\|\tilde f\|_\infty$. Because $f$ and $\tilde f$ coincide almost everywhere, if $\|\tilde f\|_\infty<\infty$, the same proof as for continuous densities gives this non-asymptotic lower bound for the length of the $k$-TSP, where $f$ has simply been replaced by $\tilde f$.

\begin{prop}
\label{prop:lower boud kTSP non uniform}
Assume $n$ vertices are drawn independently, on a compact space $\cK$, according to a density $f$ such that its Lebesgue derivative $\tilde f$ is bounded on $\cK$. Denote by $l_{TSP}(k,n)$ the length of the $k-$TSP on these $n$ vertices, where $2\leq k\leq n$. There exists a universal constant $c>0$ such that
\begin{equation*}
    \mathbb E[l_{TSP}(k,n)] \geq c\frac{k-1}{(\|\tilde f\|_\infty n)^{\frac{1}{2}\left(1+\frac{1}{k-1}\right)}} \mathcal A_\cK^{-\frac{1}{2(k-1)}}.
\end{equation*}
\end{prop}

For the upper bound, we provide similar asymptotic results, which match the lower bound whenever $k\to\infty$ and $k=o(n)$.

\begin{prop}
\label{prop:upper bound kTSP non uniform}
Assume $n$ vertices are drawn independently, on a compact space $\cK$, according to a density $f$ such that $\tilde f$ is bounded on $\cK$. Denote by $l_{TSP}(k_n,n)$ the length of the $k_n-$TSP on these $n$ vertices, where $2\leq k_n=o(n)$. There exists a universal constant $C>0$ such that
\begin{equation*}
    \limsup_{n\to\infty} \mathbb E[l_{TSP}(k_n,n)] \frac{(\|\tilde f\|_\infty n)^{\frac{1}{2}\left(1+\frac{1}{k_n-1}\right)}}{k_n-1}\psi_n\leq C.
\end{equation*}
where $\psi_n=1$ if $k_n\to\infty$ and for any sequence $\psi_n\to 0$ otherwise.
\end{prop}

\begin{proof}
We first recall that in the Lebesgue differentiation theorem, we can extend the family of balls centered at each point by families of sets with bounded eccentricity $\mathcal V$ in other words, there exists $c>0$ such that every set $U\in \mathcal V$ is contained in a ball $B$ with $|U|\geq c|B|,$ and such that every point $x$ is contained in arbitrarily small sets of the family $\mathcal V$. For instance, in this proof, we can define $\mathcal V$ as the family of cubes. The Lebesgue differential theorem gives
\begin{equation*}
    \tilde f(x) = \lim_{U\to x,U\in \mathcal V} \frac{1}{|U|}\int_U f,
\end{equation*}
where $U\to x$ means that the sets shrink to $x$ i.e. $x\in U$ and their diameters tend to $0$.

Now let $\epsilon>0$ be an error tolerance. Consider a cube $U^\epsilon$ such that
\begin{equation*}
   \left| \frac{1}{|U^\epsilon|}\int_{U^\epsilon} f - \|\tilde f\|_\infty \right|\leq \epsilon\|\tilde f\|_\infty.
\end{equation*}
For convenience, let us write $f(U^\epsilon):=\frac{1}{|U^\epsilon|}\int_{U^\epsilon} f$, and let $N(U^\epsilon)$ denote the number of vertices contained in $U^\epsilon$. According to the Hoeffding inequality, with probability $1-e^{-2\epsilon^2 f(U^\epsilon)^2 n}$, $U^\epsilon$ contains at least $n_{U^\epsilon}:=|U^\epsilon|f(U^\epsilon)(1-\epsilon) n\geq |U^\epsilon|\|\tilde f\|_\infty (1-\epsilon)^2 n$ vertices. We call $E_0$ this event. Note that $k=o(n_{U^\epsilon})$. First suppose $k\leq n^{1/3}$. Conditionally on $E_0$, these $n_{U^\epsilon}$ vertices are drawn independently according to a density $\frac{f}{|U^\epsilon|f(U^\epsilon)}$ on $U^\epsilon$. We will now focus on the $k-$TSP in $U^\epsilon$, which will serve as upper bound for the $k-$TSP on $\cK$. From here, the proof is very similar to that of the continuous density case. Let us fix $\alpha>0$. We start by partitioning $U^\epsilon$ into $P_\alpha:= m_\alpha^2$ sub-squares of equal size $\frac{\sqrt{|U^\epsilon|}}{m_\alpha}\times \frac{\sqrt{|U^\epsilon|}}{m_\alpha}$ where $m_\alpha:= \left\lfloor \frac{1}{\alpha} \sqrt{\frac{{n_{U^\epsilon}}^{1+\frac{1}{k-1}}}{k-1}}\right\rfloor$. We will show that with high probability, there exists at least one of these sub-squares that contains at least $k$ vertices. Define $X_i^\alpha$ as the number of vertices in sub-square $i$. Conditionally on $E_0$, $(X_1^\alpha,\cdots,X_{P_\alpha}^\alpha)$ follows a multinomial where the probability corresponding to sub-square $Q_i^\alpha$ is $p_i=\frac{1}{|U^\epsilon|f(U^\epsilon)}\int_{Q_i^\alpha}f \leq \frac{|Q_i^\alpha|\|\tilde f\|_\infty}{|U^\epsilon|f(U^\epsilon)}\leq \frac{1}{P_\alpha(1-\epsilon)}$. Now denote by $A_i^\alpha:=\{X_i^\alpha\geq k\}$ the event that sub-square $i$ contains at least $k$ vertices. We first give a lower bound on $\mathbb P(A_i^\alpha)$:
\begin{equation*}
    \mathbb P(A_i^\alpha)
    \geq \binom{n_{U^\epsilon}}{k}p_i^k\left(1-\frac{1}{P_\alpha(1-\epsilon)}\right)^{n_{U^\epsilon}-k}\geq \frac{n_{U^\epsilon}^k}{k!} \cdot(1+o(1))\cdot p_i^k.
\end{equation*}
By Jensen's inequality, $\frac{1}{P^\alpha}\sum_{i=1}^{P_\alpha} p_i^k\leq \frac{1}{P_\alpha^k}.$ Therefore, 
\begin{equation*}
    \sum_{i=1}^{P_\alpha} \mathbb P(A_i^\alpha) \geq c\cdot \alpha^{2k-2} \cdot(1+o(1))\geq \tilde c\cdot \alpha^{2k-2}
\end{equation*}
for some constant $\tilde c>0$, so we can use the same proof as in the case of uniform probabilities in the original paper. Then, if $l_{U^\epsilon}(k,n_{U^\epsilon})$ denotes the length of the $k-$TSP on the $n_{U^\epsilon}$ vertices in $U^\epsilon$, we obtain,
\begin{equation*}
    \mathbb E[l_{U^\epsilon}(k,n_{U^\epsilon})|E_0] \leq \hat C\frac{k-1}{n_{U^\epsilon}^{\frac{1}{2}\left(1+\frac{1}{k-1}\right)}} \sqrt{|U^\epsilon|} \leq |U^\epsilon|^{\frac{-1}{2(k_n-1)}} \frac{C_1(k_n-1)}{(\|\tilde f\|_\infty n)^{\frac{1}{2}\left(1+\frac{1}{k_n-1}\right)}},
\end{equation*}
for some constant $C_1$. If $E_0$ is not realized, we can use the naive bound $l_{TSP}(k_n,n)\leq l_{TSP}(n,n)\leq 2\sqrt n+C$. Therefore,
\begin{align*}
    \mathbb E[l_{TSP}(k_n,n)]&\leq \mathbb  E[l_{U^\epsilon}(k,n_{U^\epsilon})|E_0] + (2\sqrt n+C)(1-\mathbb P(E_0)) \\
    &\leq  |U^\epsilon|^{\frac{-1}{2(k_n-1)}} \frac{C_1(k_n-1)}{(\|\tilde f\|_\infty n)^{\frac{1}{2}\left(1+\frac{1}{k_n-1}\right)}}\cdot (1+o(1)).
\end{align*}
Note that $|U^\epsilon|^{1/(2(k_n-1))}\to 1$ if $k_n\to\infty.$ Otherwise, we can use $|U^\epsilon|^{1/(2(k_n-1))}\leq |U^\epsilon|^{-1/2}=o(\psi_n)$ for any sequence $\psi_n\to\infty$ which ends the proof for $k_n\leq n^{1/3}$. In the case where $k\geq n^{1/3}$, the same proof as in the uniform density case shows that
\begin{equation*}
    \mathbb E[l_{U^\epsilon}(k,n_{U^\epsilon})|E_0]\leq \sqrt{|U^\epsilon|}\frac{k-1}{n_{U^\epsilon}}(2\sqrt{n_{U^\epsilon}}+C)\leq 2 \frac{k-1}{(\|\tilde f\|_\infty n)^{\frac{1}{2}\left(1+\frac{1}{k_n-1}\right)}}(1+o(1)).
\end{equation*}
The proof follows from the same arguments as in the case $k_n\leq n^{1/3}$.
\end{proof}

For the case $k=\Theta(n)$, we expect a constant-factor approximation for the $k-$TSP to perform the TSP on a set with maximal average density and area $\Theta(k/n)$. In the following, we state this generalization as a claim without proof. A possible proof sketch would use similar techniques to the analysis developed for the TRP in the original paper.

\begin{claim}Assume $n$ vertices are drawn independently, on a compact space $\cK$, according to a density $f$. Let $\epsilon>0$. Denote by $l_{TSP}(k_n,n)$ the length of the $k_n-$TSP on these $n$ vertices, where $\epsilon n\leq k_n \leq n$. There exists constants $0<c_\epsilon<C$ such that
\begin{equation*}
    c_\epsilon \leq \liminf_{n\to\infty}  \frac{\mathbb  E[l_{TSP}(k_n,n)]}{\sqrt n g_f(k_n/n)}\leq \limsup_{n\to\infty} \frac{\mathbb  E[l_{TSP}(k_n,n)]}{\sqrt n g_f(k_n/n)}\leq C,
\end{equation*}
where if we denote by $F$ the cumulative distribution of $f$ and $y_0=\inf\{y: 1-F(y)\leq k_n/n\}$,
\begin{equation*}
    g_f(k_n/n) = \int \sqrt{f}\mb 1_{f>y_0} + \frac{k_n/n-(1-F(y_0))}{\sqrt{y_0}}.
\end{equation*}
\end{claim}

\section{The $\Psi$-TRP}
\label{sec:psi-TRP}

In the main paper, we analyzed the TRP under fairness considerations. In particular, we showed that achieving efficency while ensuring max-min fairness asymptotically is possible. Here, we propose another notion of fairness and give similar positive results. Recall that the TRP objective of a given tour is
\begin{equation*}
    \sum_{i=1}^n l_i,
\end{equation*}
where $l_i$ is the latency at vertex $i$. In resource allocation problems, this objective corresponds to the utilitarian principle i.e. maximizing the total utility. A common approach to fairness consists of maximizing $\sum_j f(u_j)$ where $f$ is a concave function and $u_j$ denotes the utility of player $j$. In particular, the $\log$ function yields the proportional fairness solution under mild convexity assumptions \citep{bertsimas2011price}. We adapt this idea to our setting by changing the latency objective. Specifically, for any increasing function $\Psi$, we can define the $\Psi-$TRP, which seeks a tour that minimizes the objective:
\begin{equation*}
    \sum_{i=1}^n \Psi(l_i).
\end{equation*}
To capture fairness considerations, we assume that $\Psi$ is convex. We show that, for a large class of functions $\Psi$, our approximation algorithm for the TRP is also constant-factor optimal for the $\Psi$-TRP, hence encapsulating this notion of fairness. Indeed, our analysis for the TRP generalizes to the $\Psi-$TRP when $\Psi$ is a convex monomial. This is formalized in the following proposition, which we prove in the next sections.

\begin{prop}
\label{prop:psi-TRP objective}
Assume all $n$ vertices are drawn according to a distribution with density $f$ on a compact space $\cK\subset \mathbb R^2$. Let $\alpha\geq 1$ and $\Psi_\alpha:x\mapsto x^\alpha$ the power function. Denote by $l_{\alpha-TRP}$ the optimal $\Psi_\alpha-$TRP objective of a tour for the $\Psi_\alpha-$TRP. Then,
\begin{equation*}
    c_\alpha\int_\cK g_\alpha(f,x) dx \leq \liminf_{n\to\infty} \frac{\mathbb E\left[ l_{\Psi-TRP}\right]}{n^{1+\alpha/2}} \leq \limsup_{n\to\infty} \frac{\mathbb E\left[ l_{\Psi-TRP}\right]}{n^{1+\alpha/2}} \leq C_\alpha\int_\cK g_\alpha(f,x) dx
\end{equation*}
where $0<c_\alpha<C_\alpha$ are two constants depending only in $\alpha$ and
\begin{equation*}
    g_\alpha(f,x) = \begin{cases} 
          f(x)\left(\int_\cK\sqrt f \cdot \mb 1_{f< f(x)}\right)^\alpha, & \text{if }\int_\cK \mb 1_{f=f(x)}=0 \\
          \sqrt{f(x)}\frac{\left(\int_\cK\sqrt{f}\cdot \mb 1_{f\leq f(x)}\right)^{\alpha+1}-\left(\int_\cK\sqrt{f}\cdot \mb 1_{f<f(x)}\right)^{\alpha+1}}{(\alpha+1)\int_\cK \mb 1_{f=f(x)}} , & \text{otherwise}.
       \end{cases}
\end{equation*}
\end{prop}

We use similar proof ideas as for the probabilistic bounds of the classical TRP. However, because $\Psi_\alpha$ is non-linear for $\alpha>1$, the arguments are more technical. In particular for the lower bound, we divide the tour into sub-paths in each sub-square of the partition but with the additional constraint that all sub-paths should visit the same number of vertices. The non-linearity of $\Psi_\alpha$ also affects the form of the integrand $g_\alpha(f,\cdot)$ for degenerate levels of the density function when $\int_\cK\mb 1_{f=f(x)}>0$. As a result, the proof of convergence of the integral of $g_\alpha(\phi,\cdot)$ to the integral of $g_\alpha(f,\cdot)$, for fine piece-wise constant approximations $\phi$ of $f$, is more technical than the equivalent result for the TRP.

Furthermore, the upper bound is reached by the same approximating scheme as for the TRP in which we partition the space in sub-squares and visit sub-squares by decreasing order of density. In particular, this scheme is also constant-factor optimal for the $\Psi-$TRP. Using the same arguments, we can generalize Proposition~\ref{prop:psi-TRP objective} to any linear combination of monomials where the leading term is a of the form $x\mapsto c\cdot x^\alpha$ where $c>0$ and $\alpha\geq 1$. In other words, the competitive ratio between the fairness-maximizing $\Psi-$TRP and the efficiency-maximizing TRP is asymptotically $1$.

\subsection{A lower bound}

\begin{prop}
\label{prop:lower bound Psi TRP}
Assume all $n$ vertices are drawn according to a distribution with density $f$ on a compact space $\cK\subset \mathbb R^2$. Let $\alpha\geq 1$ and $\Psi_\alpha:x\mapsto x^\alpha$ the power function. Denote by $l_{\alpha-TRP}$ the optimal TRP objective of a tour for the $\Psi_\alpha-$TRP. Then,
\begin{equation*}
    \liminf_{n\to\infty} \frac{\mathbb E\left[ l_{\Psi-TRP}\right]}{n^{1+\alpha/2}} \geq c_\alpha\int_\cK g_\alpha(f,x)dx,
\end{equation*}
where $c_\alpha:=\frac{1}{(\pi e)^{\alpha/2}} $ is a constant and
\begin{equation*}
    g_\alpha(f,x) = \begin{cases} 
          f(x)\left(\int_\cK\sqrt f \cdot \mb 1_{f< f(x)}\right)^\alpha, & \text{if }\int_\cK \mb 1_{f=f(x)}=0 \\
          \sqrt{f(x)}\frac{\left(\int_\cK\sqrt{f}\cdot \mb 1_{f\leq f(x)}\right)^{\alpha+1}-\left(\int_\cK\sqrt{f}\cdot \mb 1_{f<f(x)}\right)^{\alpha+1}}{(\alpha+1)\int_\cK \mb 1_{f=f(x)}} , & \text{otherwise,}
       \end{cases}
\end{equation*}
is a function that depends only on $\alpha$ and $f$.
\end{prop}

\begin{proof}
We take the same notations as in the proof of the lower bound of Theorem 3 from \citep{blanchard2021probabilistic}. Again, we first start by the case where $f$ has support in the unit square $[0,1]^2$ and has the form
\begin{equation*}
    f = \sum_{k=1}^{m^2}f_k \mb 1_{Q_k},
\end{equation*}
where $\{Q_i\}$ is the regular partition of the unit square into $m^2$ sub-squares. We define the margin
\begin{equation*}
    \mathcal M = \bigcup_{1\leq k\leq m^2} Q_k\cap\left(\partial Q_k + \epsilon_m^{(k)} B(0,1)\right),
\end{equation*}
for $\epsilon^{(k)}_m := \frac{\epsilon}{m}\sqrt{\frac{f_*}{f_k}}.$ Note that this is a smaller margin than what was considered in the proof of the lower bound of Theorem 3 from \citep{blanchard2021probabilistic}. We can have estimates for the number of vertices in the margin similar to Lemma 3 of \citep{blanchard2021probabilistic}. Finally, we define the event $E_0$ in which for all $1\leq k\leq m^2$ such that $f_k>0$,
\begin{equation*}
     \frac{f_k}{2m^2}n\leq N_k\leq \frac{3f_k}{2m^2}n, \quad l_{TSP(Q_k)}\left(\left\lceil \epsilon\cdot e\sqrt{\pi\frac{3f_*}{2m^2}n}\right\rceil,N_k\right)> \epsilon^{(k)}_n.
\end{equation*}
Let us estimate the probability of the event $E_0$. By the proof of Lemma 4 of \citep{blanchard2021probabilistic},
\begin{equation*}
    \mathbb P\left[\left|N_k- \frac{f_k}{m^2}n\right|\geq \frac{f_k}{2m^2}n\right] \leq 2 e^{-\frac{1}{12}\cdot \frac{f_k}{m^2}n}.
\end{equation*}
We now use Corollary 1 of \citep{blanchard2021probabilistic} to each of the sub-squares. For $1\leq k \leq m^2$, such that $f_k>0$,
\begin{align*}
    \mathbb P \left[\left. l_{TSP(Q_k)}\left(\left\lceil \epsilon\cdot e\sqrt{\pi\frac{3f_*}{2m^2}n}\right\rceil,N_k\right)\leq \epsilon_m^{(k)} \right|\right.&\left.\frac{f_k}{2m^2}n< N_k< \frac{3f_k}{2m^2}n\right]\\
    &\leq \mathbb P \left[ l_{TSP(Q_k)}\left(\left\lceil\epsilon\cdot e\sqrt{\pi\frac{3f_*}{2m^2}n}\right\rceil,\frac{3f_k}{2m^2}n\right)\leq \epsilon_m^{(k)}\right]\\
    &=o\left(e^{-\epsilon\cdot e \sqrt{\pi\frac{3f_* n}{2m^2}}}\right).
\end{align*}
Finally, the probability of $E_0$ is $1-o\left(e^{-c\epsilon \sqrt{\frac{f_* n}{m^2}}}\right)$ for some constant $c>0$.

In the next steps we will assume that this event is met. We are now ready to use an equivalent of Lemma 5 from \citep{blanchard2021probabilistic} to each sub-path in $Q_k$ which is not completely included in the margin. However, we will need all sub-paths to visit same number of vertices. Denote by $l_{\Psi-TRP}$ the optimal objective and consider an optimal tour. We order the sub-paths $\mathcal P_1,\cdots \mathcal P_P$ which are not completely included in the margin. Also, we denote by $k(i)$ the index of the sub-square containing $\mathcal P_i$. We divide $\mathcal P_i$ into smaller sub-paths of length exactly $n_* = \left\lceil \epsilon\cdot e\sqrt{\pi\frac{3f_*}{2m^2}n}\right\rceil$. Since the number of vertices visited by $\mathcal P_i$ might not be a multiple of $n_*$, some vertices will be left out. For any path $\mathcal P_i$, if $n(\mathcal P_i)\geq n_*$, then at most $n(\mathcal P_i)/2$ vertices will be left out. We will denote $\mathcal P_i^1,\cdots \mathcal P_i^{t_i}$ the corresponding created sub-paths containing exactly $n_*$ edges. Note that $t_i = \left\lfloor \frac{n(\mathcal P_i)}{n_*}\right\rfloor.$ We now treat paths with $n(\mathcal P_i)<n_*$ separately, which we will call low-density paths. Let $L$ be the set of indices of low density paths. For a given low-density path $\mathcal P_i$, we artificially add $n_*-n(\mathcal P_i)$ vertices from later low-density paths $\mathcal P_j$ in the same sub-square as $\mathcal P_i$, where $j>i$. At the end of this process, at most one low-density path remains, which we will leave out. Let us denote by $\hat {\mathcal P}_i$ for $i\in \tilde L$, the corresponding constructed paths from low-density paths. Note that we have $\tilde L\subset L$, but not necessarily an equality because the process can potentially remove all vertices of some low-density paths. In the following, if $\mathcal P$ is a sub-path, we will denote by $l(\mathcal P)$ its length. Let us summarize the obtained lower bound.
\begin{align*}
    l_{\Psi-TRP} &= \sum_{v\in V} \Psi_\alpha(\text{completion time of } v)\\
    &= \sum_{1\leq i\leq P} \sum_{v\in \mathcal P_i} \Psi_\alpha(\text{completion time of } v)\\
    &\geq \sum_{i\notin L}\sum_{t=1}^{t_i} \sum_{v\in \mathcal P_i^t} \Psi_\alpha(\text{completion time of } v) + \sum_{i\in \tilde L} \sum_{v\in \hat {\mathcal P}_i} \Psi_\alpha(\text{completion time of } v)\\
    &\geq \sum_{i\notin L}\sum_{t=1}^{t_i} \sum_{v\in \mathcal P_i^t} \Psi_\alpha\left(\sum_{j=1}^{i-1} l(\mathcal P_j) + \sum_{u=1}^{t-1} l(\mathcal P_i^u)\right) + \sum_{i\in \tilde L} \sum_{v\in \hat {\mathcal P}_i} \Psi_\alpha\left(\sum_{j=1}^{i-1} l(\mathcal P_j)\right)\\
    &= n_*\left[ \sum_{i\notin L}\sum_{t=1}^{t_i} \Psi_\alpha\left(\sum_{j=1}^{i-1} l(\mathcal P_j) + \sum_{u=1}^{t-1} l(\mathcal P_i^u)\right) + \sum_{i\in \tilde L} \Psi_\alpha\left(\sum_{j=1}^{i-1} l(\mathcal P_j)\right)\right]\\
    &\geq n_*\sum_{1\leq i\leq \tilde P} \Psi_\alpha\left(\sum_{j=1}^{i-1} l(\tilde {\mathcal P}_j)\right),
\end{align*}
where we have listed the new sub-paths containing $n_*$ vertices: $\tilde {\mathcal P}_1,\cdots \tilde {\mathcal P}_{\tilde P}$ with the order given by the original tour --- the ordering where we omit added vertices to low-density sub-paths. The length of the subpath $\tilde {\mathcal P}_i$ is the length of the corresponding subpath $\tilde {\mathcal P}_j^t$ if it came from a non low-density path. Otherwise, we define it as $l(\tilde {\mathcal P}_i) := l({\mathcal P}_j)$ where $\tilde {\mathcal P}_i=\hat {\mathcal P}_j$. This corresponds to lower bounding the contribution of added vertices in low-density sub-paths, to the $\Psi-TRP$ objective. A key observation is that we can have a similar result to that of Lemma 5 from \citep{blanchard2021probabilistic}. Again, we will denote by $k(i)$ the index of the sub-square containing sub-path $\tilde {\mathcal P}_i$, i.e. $\tilde {\mathcal P}_i\subset Q_{k(i)}$.

\begin{lemma}
\label{lemma:lower bound length subpath Psi-TRP}
Let $1\leq i\leq \tilde P$. Under the event $E_0$, for $n$ sufficiently large, we can give the lower bound
\begin{equation*}
    l(\tilde {\mathcal P}_i)\geq \frac{n_*}{2\sqrt{\pi e\cdot f_{k(i)} n}}.
\end{equation*}
\end{lemma}

\begin{proof}
Under $E_0$, no path containing at $n_*$ vertices has lower length than $\epsilon_m^{(k)}$. Let us first consider the case of a sub-path $\tilde p_i$ corresponding to a sub-path of a non low-density sub-path $p_j$. Then, $\tilde p_i$ is a ``true" sub-path of the original tour and contains $n_*$ vertices. Therefore, $l_{\tilde p_i}\geq \epsilon_m^{(k(i))}$. Let us now consider a sub-path $\tilde p_i$ corresponding  to a low-density sub-path $p_j$ for $j\in L$. Recall that $p_j$ is a sub-path of $Q_{k(i)}$ which is not entirely contained in the margin. Therefore, $l_{p_j}\geq \epsilon_m^{(k(i))}.$ In summary, for all $1\leq i\leq \tilde P$,
\begin{equation*}
    l_{\tilde p_i}\geq \epsilon_m^{(k(i))} \geq \frac{n_*}{2\sqrt{\pi e\cdot f_{k(i)} n}},
\end{equation*}
where the second inequality is true for $n$ sufficiently large.
\end{proof}
Therefore, under $E_0$ we have the following lower bound,
\begin{align*}
    l_{\Psi-TRP} &\geq n_*\sum_{1\leq i\leq \tilde P} \Psi_\alpha\left(\sum_{j=1}^{i-1} l_{\tilde p_j}\right)\\
    &\geq n_*\sum_{1\leq i\leq \tilde P} \Psi_\alpha\left(\frac{1}{2\sqrt{\pi e n}}\sum_{j=1}^{i-1} \frac{n_*}{\sqrt{f_{k(j)}}}\right)\\
    &= \frac{n_*^{\alpha+1}}{2^\alpha(\pi e n)^{\alpha/2}}\sum_{1\leq i\leq \tilde P} \Psi_\alpha\left(\sum_{j=1}^{i-1}  \frac{1}{\sqrt{f_{k(j)}}}\right)\\
    &\geq \frac{n_*^{\alpha+1}}{2^\alpha(\pi e n)^{\alpha/2}}\cdot \min_{\sigma \in \mathcal S_{\tilde P}} \sum_{i} \Psi_\alpha\left(\sum_{j=1}^{i-1}  \frac{1}{\sqrt{f_{k(\sigma(j))}}}\right).
\end{align*}
Let us now give an equivalent of Lemma 6 from \citep{blanchard2021probabilistic}.
\begin{lemma}
The minimum objective of the optimization problem
\begin{equation*}
    \min_{\sigma \in \mathcal S_{\tilde P}} \sum_{i} \Psi_\alpha\left(\sum_{j=1}^{i-1}  \frac{1}{\sqrt{f_{k(\sigma(j))}}}\right).
\end{equation*}
is given by ordering sub-paths by increasing order of $\frac{1}{\sqrt{f_{k(i)}}}$, i.e. decreasing order of $f_{k(i)}$.
\end{lemma}

\begin{proof}
In this proof, we will denote by $C_\sigma$ the objective of the minimization problem for $\sigma\in\mathcal S_P$, i.e.
\begin{equation*}
    C_\sigma := \min_{\sigma \in \mathcal S_{\tilde P}} \sum_{i} \Psi_\alpha\left(\sum_{j=1}^{i-1}  \frac{1}{\sqrt{f_{k(\sigma(j))}}}\right).
\end{equation*}
Let $1\leq i<\tilde P$. We will compare $C_\sigma$ and $C_{\tilde \sigma}$ where $\tilde \sigma$ is the permutation $\sigma$ but the $i-$th and $(i+1)-$th index are interchanged:
\begin{equation*}
    \tilde\sigma(r) = \begin{cases} 
          \sigma(r), & r\notin\{i,i+1\} \\
          \sigma(i+1), & r=i \\
          \sigma(i), & r=i+1.
       \end{cases}
\end{equation*}
Then,
\begin{align*}
    C_{\tilde \sigma}-C_\sigma &= \Psi\left(\frac{1}{\sqrt {f_{k(\sigma(i+1))}}} +  \eta\right) - \Psi\left(\frac{1}{\sqrt {f_{k(\sigma(i))}}} +  \eta\right),
\end{align*}
where $\eta = \sum_{t<i}\frac{n_{p_{\sigma(t)}}}{\sqrt {f_{k(\sigma(t))}}}$. Assume that we have $\frac{1}{\sqrt {f_{k(\sigma(i+1))}}} \leq \frac{1}{\sqrt {f_{k(\sigma(i))}}}$. Then, the objective is decreases when we place $\sigma(i+1)$ in $i-$th position: $C_{\tilde \sigma} \leq C_\sigma.$ We then use this argument to order sequentially the permutation $\sigma$ by decreasing order of $f_{k(i)}$. This ends the proof of the lemma.
\end{proof}

Let us now give estimates on the right hand of the inequality. Denote by $\sigma^*$ the ordering on the sub-squares $Q_k$ such that $\frac{1}{\sqrt{f_{\sigma^*(k)}}}$ is increasing in $k$. Then under $E_0$,
\begin{align*}
    \min_{\sigma \in \mathcal S_{\tilde P}}& \sum_{i} n_*\Psi_\alpha\left(\sum_{j=1}^{i-1} \frac{n_*}{\sqrt {f_{k(\sigma(j))}}}\right) \\
    &\geq \sum_{1\leq k\leq m^2} \left(\sum_{i,\; \tilde p_i\in Q_{\sigma^*(k)}} n_* \right)\cdot \Psi_\alpha\left[\sum_{t=1}^{k-1} \frac{1}{\sqrt {f_{\sigma^*(t)}}}\left(\sum_{i,\; \tilde p_i\in Q_{\sigma^*(t)}} n_*\right)\right]\\
    &\geq \sum_{1\leq k\leq m^2} \frac{N_{\sigma^*(k)} - |V\cap Q_{\sigma^*(k)}\cap \mathcal M|-n_*}{2}  \cdot\Psi_\alpha\left[\sum_{t=1}^{k-1} \frac{N_{\sigma^*(t)} - |V\cap Q_{\sigma^*(t)}\cap \mathcal M|-n_*}{2\sqrt {f_{\sigma^*(t)}}}\right]\\
    &\geq \frac{1}{2^{\alpha+1}}\sum_{1\leq k\leq m^2} N_{\sigma^*(k)} \cdot\Psi_\alpha\left[\sum_{t=1}^{k-1} \frac{N_{\sigma^*(t)}}{\sqrt {f_{\sigma^*(t)}}}\right] -\frac{1}{2^{\alpha+1}}\sum_{1\leq k\leq m^2} (|V\cap Q_{\sigma^*(k)}\cap \mathcal M|+n_*) \cdot\Psi_\alpha\left[\frac{n}{\sqrt {f_*}}\right]\\
    &\quad\quad\quad\quad-\frac{\alpha}{2^{\alpha+1}}\sum_{1\leq k\leq m^2} N_{\sigma^*(k)}\sum_{t=1}^{k-1}\frac{|V\cap Q_{\sigma^*(t)}\cap \mathcal M|+n_*}{\sqrt{f_*}} \cdot\Psi_{\alpha-1}\left[\frac{n}{\sqrt {f_*}}\right]\\
    &\geq \frac{n^{\alpha+1}}{4^{\alpha+1}m^{2\alpha+2}}\sum_{1\leq k\leq m^2} f_{\sigma^*(k)} \cdot\Psi_\alpha\left[\sum_{t=1}^{k-1} \sqrt {f_{\sigma^*(t)}}\right] -\frac{|V\cap\mathcal M|+m^2n_*}{2^{\alpha+1}} \cdot\Psi_\alpha\left[\frac{n}{\sqrt {f_*}}\right]\\
    &\quad\quad\quad\quad-\frac{\alpha \cdot n}{2^{\alpha+1}}\frac{|V\cap \mathcal M|+m^2n_*}{\sqrt{f_*}} \cdot\Psi_{\alpha-1}\left[\frac{n}{\sqrt {f_*}}\right]\\
    &= \frac{n^{\alpha+1}}{4^{\alpha+1}m^{2\alpha+2}}\sum_{1\leq k\leq m^2} f_{\sigma^*(k)} \cdot\Psi_\alpha\left[\sum_{t=1}^{k-1} \sqrt {f_{\sigma^*(t)}}\right] -\frac{n^\alpha (1+\alpha)}{2^{\alpha+1}f_*^{\alpha/2}}(|V\cap \mathcal M|+m^2n_*).
\end{align*}
Using Lemma 3 of \citep{blanchard2021probabilistic}, we obtain that with probability $1-o(e^{-c\epsilon \sqrt{\frac{f_* n}{m^2}}})$, the event $E_0$ is met and $|V\cap\mathcal M|\leq 8\epsilon n$. Therefore, we can take $\epsilon>0$ sufficiently small so that
\begin{equation*}
    l_{TRP} \geq \frac{1+o(1)}{8^{1+\alpha} (\pi e)^{\alpha/2}m^{2\alpha+2}} n^{1+\alpha/2}\sum_{1\leq k\leq m^2} f_{\sigma^*(k)} \cdot\Psi_\alpha\left[\sum_{t=1}^{k-1} \sqrt {f_{\sigma^*(t)}}\right] .
\end{equation*}

Define a new constant $c_\alpha :=\frac{1}{8^{1+\alpha} (\pi e)^{\alpha/2}}$, we now obtain the desired result.
\begin{align*}
  \liminf_{n\to\infty} \frac{\mathbb E\left[ l_{\Psi-TRP}\right]}{n^{1+\alpha/2}} &\geq  \liminf_{n\to\infty} \mathbb P(E_0)\cdot \frac{\mathbb E[l_{\Psi-TRP}|E_0]}{n^{1+\alpha/2}}\\
  &\geq  \liminf_{n\to\infty}  \left(1-o\left(e^{-c\epsilon \sqrt{\frac{f_* n}{m^2}}}\right)\right) \frac{c_\alpha(1+o(1))}{m^{2\alpha+2}}\sum_{1\leq k\leq m^2} f_{\sigma^*(k)} \Psi_\alpha\left[\sum_{t=1}^{k-1} \sqrt {f_{\sigma^*(t)}}\right]\\
  &\geq \frac{c_\alpha}{m^{2\alpha+2}}\sum_{1\leq k\leq m^2} f_{\sigma^*(k)} \cdot\Psi_\alpha\left[\sum_{t=1}^{k-1} \sqrt {f_{\sigma^*(t)}}\right] .
\end{align*}

We will now make the link between the discrete sum and the integral formula. To do so, we aggregate sub-squares who have same density $f_k$. If $f^1> \cdots > f^S$ are values taken by the density function, We obtain a partition $\{1,\cdots,m^2\} = \bigcup_{1\leq s\leq S} A_s,$ where $A_s=\{k:\; f_k=f^s\}$ contains the indices of sub-squares having density $f^s$. Note that because the values $f^1,\cdots,f^S$ are ordered, so are the sets $A_s$ i.e. all elements of $A_2$ are larger than elements of $A_1$, etc. Then,
\begin{align*}
    &\frac{1}{m^{2\alpha+2}}\sum_{1\leq k\leq m^2} f_{\sigma^*(k)} \cdot\Psi_\alpha\left[\sum_{t=1}^{k-1} \sqrt {f_{\sigma^*(t)}}\right]\\
    &= \sum_{1\leq k\leq m^2}\int_{Q_{\sigma^*(k)}} f(x) \left(\int_\cK  \sqrt f\cdot \mb  1_{Q_{\sigma^*(1)}\cup\cdots\cup Q_{\sigma^*(k-1)}} \right)^\alpha dx\\
    &= \sum_{s=1}^S \sum_{k\in A_s} \int_{Q_{\sigma^*(k)}} f(x) \left(\int_\cK \sqrt f \cdot \mb 1_{f<f(x)} + \sqrt {f(x)}\cdot \mb 1_{\bigcup_{l\in A_s, l<k} Q_{\sigma^*(k)}} \right)^\alpha dx.
\end{align*}
Therefore, 
\begin{align*}
    \sum_{s=1}^S \int_{0}^{\mathcal A\left(\bigcup_{l\in A_s} Q_{\sigma^*(k)}\right)} &f(x) \left(\int_\cK \sqrt f \cdot \mb 1_{f<f(x)} + t\sqrt {f(x)}\right)^\alpha dt - \frac{1}{m^{2\alpha+2}}\sum_{1\leq k\leq m^2} f_{\sigma^*(k)} \cdot\Psi_\alpha\left[\sum_{t=1}^{k-1} \sqrt {f_{\sigma^*(t)}}\right] \\
    &= \sum_{s=1}^S \sum_{k\in A_s} f(x) \int_0^{\mathcal A(Q_{\sigma^*(k)})}\left[ \left(\int_\cK \sqrt f \cdot \mb 1_{f<f(x)} + \sqrt {f(x)}\cdot \mathcal A(\cup_{l\in A_s, l<k} Q_{\sigma^*(k)}) + t\sqrt{f(x)} \right)^\alpha\right.\\
    &\quad \quad \quad - \left. \left(\int_\cK \sqrt f \cdot \mb 1_{f<f(x)} + \sqrt {f(x)}\cdot \mathcal A(\cup_{l\in A_s, l<k} Q_{\sigma^*(k)}) \right)^\alpha \right] dt\\
    &\leq \sum_{s=1}^S \sum_{k\in A_s} f(x) \int_0^{\mathcal A(Q_{\sigma^*(k)})} \alpha \left(\int_\cK \sqrt f\right)^{\alpha-1} t\sqrt{f(x)}dt\\
    &=  \frac{\alpha}{2m^2} \left(\int_\cK \sqrt f\right)^{\alpha-1} \int_\cK f^{3/2}.
\end{align*}
Also note that
\begin{equation*}
    \sum_{s=1}^S \int_{0}^{\mathcal A\left(\bigcup_{l\in A_s} Q_{\sigma^*(k)}\right)} f(x) \left(\int_\cK \sqrt f \cdot \mb 1_{f<f(x)} + t\sqrt {f(x)}\right)^\alpha dt
    = \int_\cK g_\alpha(f,x)dx.
\end{equation*}
Finally, we have
\begin{equation*}
  \liminf_{n\to\infty} \frac{\mathbb E\left[ l_{\Psi-TRP}\right]}{n^{1+\alpha/2}}
  \geq c_\alpha \int_\cK g_\alpha(f,x) dx  - \frac{c_\alpha \alpha}{2m^2} \left(\int_\cK \sqrt f\right)^{\alpha-1} \int_\cK f^{3/2}.
\end{equation*}
We can repeat the same procedure with a finest partition of the unit square $[0,1]^2$ into $(\beta m)^2$ sub-squares where $\beta\in \mathbb N^*$. For $\beta$ sufficiently large, we obtain
\begin{equation*}
    \frac{\alpha}{2m^2} \left(\int_\cK \sqrt f\right)^{\alpha-1} \int_\cK f^{3/2}\leq \frac{1}{2}\int_\cK g_\alpha(f,x) dx.
\end{equation*}
Then, with this partition we obtain the desired result
\begin{equation*}
    \liminf_{n\to\infty} \frac{\mathbb E\left[ l_{\Psi-TRP}\right]}{n^{1+\alpha/2}}
  \geq \tilde c_\alpha \int_\cK g_\alpha(f,x) dx,
\end{equation*}
where $\tilde c_\alpha = \frac{c_\alpha}{2(\alpha+1)} = \frac{1}{2\cdot 8^{\alpha+1} (\pi e)^{\alpha/2}(\alpha+1)}$. This ends the proof for the densities of the form
\begin{equation*}
    f(x) = \sum_{k=1}^{m^2}f_k \mb{1}_{Q_k}(x).
\end{equation*}
Note that with the same proof, we can tighten the constant $\tilde c_\alpha$ to be $\frac{1}{(\pi e)^{\alpha/2}(\alpha+1)}$.

We now turn to general distributions with continuous densities. To do so, we need an equivalent of Lemma 7 from \citep{blanchard2021probabilistic}, which is given by Lemma~\ref{lemma:piece-wise constant approx psi-trp}. Similarly to the proof of the lower bound of Theorem 3 of \citep{blanchard2021probabilistic}, let us now consider the general case of an absolutely continuous density $f$ on a compact space $\cK$. By a scaling argument, we can suppose without loss of generality that $\cK\subset[0,1]^2$. For any $\epsilon>0$, we use Lemma~\ref{lemma:piece-wise constant approx psi-trp} to take a density $\phi$ of the same form as above
\begin{equation*}
    \phi(x) = \sum_{k=1}^{m^2}\phi_k \mb{1}_{Q_k}(x),
\end{equation*}
such that $\|\phi-f\|_\infty \leq \epsilon$ and $ \left|\int_\cK g_\alpha(\phi)-g_\alpha(f) \right| \leq \epsilon.$ By a coupling argument, we can construct a joint distribution $(X,Y)$ such that $X$ (resp. $Y$) has density $f$ (resp. $\phi$), and $\mathbb P(X\neq Y)\leq 2\int_\cK |\phi(x)-f(x)|dx \leq 2\epsilon.$ On the event $\{X_i=Y_i, 1\leq i\leq n\}$, the $\Psi-$TRP lengths coincide. Therefore, we can use the estimates on $\phi$ to show that
\begin{align*}
    \liminf_{n\to\infty}\frac{\mathbb E[l_{TRP}(f)]}{n^{\alpha/2}}
    &\geq (1-2\epsilon)^n c_\alpha \int_\cK g_\alpha(\phi)\\
    &\geq (1-2\epsilon)^n c_\alpha \left(\int_\cK g_\alpha(f) - \epsilon\right).
\end{align*}
Since this is valid for any $\epsilon>0$, the desired result follows.
\begin{equation*}
    \liminf_{n\to\infty}\frac{\mathbb E[l_{TRP}(f)]}{n^{\alpha/2}} \geq c_\alpha \int_\cK g_\alpha(f)
\end{equation*}
This ends the proof of the Proposition.
\end{proof}

\subsection{An upper bound}
We now give a constructive proof of an upper bound. The resulting constructed tour is constant-factor from the optimal $\Psi-$TRP tour.

\begin{prop}
\label{prop:upper bound Psi TRP}
Assume all $n$ vertices are drawn according to a distribution with density $f$ on a compact space $\cK\subset \mathbb R^2$. Let $\alpha\geq 1$ and $\Psi_\alpha:x\mapsto x^\alpha$ the power function. Denote by $l_{\alpha-TRP}$ the optimal TRP objective of a tour for the $\Psi_\alpha-$TRP. Then,
\begin{equation*}
    \liminf_{n\to\infty} \frac{\mathbb E\left[ l_{\Psi-TRP}\right]}{n^{1+\alpha/2}} \leq C_\alpha\int_\cK g_\alpha(f,x)dx,
\end{equation*}
where $C_\alpha>0$ is a constant depending only on $\alpha.$ Furthermore, there exists a simple way to construct a tour that achieves the provided upper bound.
\end{prop}

\begin{proof}
Let $\epsilon>0$. Take $m>0$ and a piece-wise constant density $\phi$ approximating $f$, given by Lemma~\ref{lemma:piece-wise constant approx psi-trp}. Similarly to the tour constructed in the proof of the upper bound of Theorem 3 of \citep{blanchard2021probabilistic}, if we order the sub-squares by decreasing value of $\phi$ and denote $\sigma$ this ordering, our tour will follow a TSP tour on $Q_{\sigma(1)}$, then on $Q_{\sigma(2)}$, until $Q_{\sigma(m^2)}$. We will now show that this tour is constant-factor optimal on the high-event probability $E_0$ in which
\begin{equation*}
\frac{\phi_k}{2m^2}n\leq N_k\leq \frac{3\phi_k}{2m^2}n,
\end{equation*}
for all $1\leq k\leq m^2$, where $N_k$ is the count of vertices in sub-square $Q_k$. By the Chernoff bound, $\mathbb P(E_0) = 1-o\left(e^{-c\frac{\phi_* n}{m^2}}\right),$ where $\phi_*:=\min\{\phi_k\}$ and $c>0$ a constant. Let us now analyze the $\Psi-$TRP objective of this tour on $E_0$. On each sub-square, by the BHH theorem, the length $l^k_{TSP}$ of the optimal TSP satisfies
\begin{equation*}
l_{TSP}^k\leq C\beta_{TSP}\sqrt{\frac{3\phi_k n}{2}}\frac{1}{m^2}.
\end{equation*}
for $C>0$ a constant and any $n$ sufficiently large. Then, if $\hat l_{\Psi-TRP}$ denotes the objective of the constructed tour, for $n$ sufficiently large,
\begin{align*}
    \hat l_{\Psi-TRP} &\leq \sum_{k=1}^{m^2} N_{\sigma(k)}\Psi_\alpha\left( \sum_{l=1}^k l_{TSP}^{\sigma(l)} + (k-1)\sqrt{2}\right)\\
    &\leq \frac{3n}{2m^2}\sum_{k=1}^{m^2} \phi_{\sigma(k)} \Psi_\alpha\left( C\beta_{TSP}\sqrt{\frac{3n}{2}}\frac{1}{m^2} \sum_{l=1}^k \sqrt{\phi_{\sigma(l)}}+ (k-1)\sqrt{2}\right)\\
    &\leq \left(\frac{3}{2}\right)^{1+\alpha/2}C^\alpha\beta_{TSP}^\alpha \frac{n^{1+\alpha/2}}{m^2}\sum_{k=1}^{m^2} \phi_{\sigma(k)} \left[\Psi_\alpha\left(\frac{1}{m^2}\sum_{l=1}^k \sqrt{\phi_{\sigma(l)}}\right) + \frac{2\alpha k}{C\beta_{TSP}\sqrt{3n}}\cdot 2\left\|\sqrt \phi\right\|_1^{\alpha-1}\right]\\
    &\leq \left(\frac{3}{2}\right)^{1+\alpha/2}C^\alpha\beta_{TSP}^\alpha \frac{n^{1+\alpha/2}}{m^{2+2\alpha}}\sum_{k=1}^{m^2} \phi_{\sigma(k)} \Psi_\alpha\left(\sum_{l=1}^k \sqrt{\phi_{\sigma(l)}}\right) + \left(\frac{3}{2}\right)^{1+\alpha/2}(C\beta_{TSP})^{\alpha-1} \frac{4\alpha m^2 n^{(\alpha+1)/2}}{\sqrt 3}.
\end{align*}
Therefore, with $C_\alpha := \left(\frac{3}{2}\right)^{1+\alpha/2}C^\alpha\beta_{TSP}^\alpha$, we obtain
\begin{align*}
  \liminf_{n\to\infty} \frac{\mathbb E\left[ l_{\Psi-TRP}\right]}{n^{1+\alpha/2}} 
  &\leq \frac{C_\alpha}{m^{2\alpha+2}}\sum_{1\leq k\leq m^2} \phi_{\sigma(k)} \cdot\Psi_\alpha\left[\sum_{t=1}^k \sqrt {\phi_{\sigma(t)}}\right]\\
  &\leq \frac{C_\alpha}{m^{2\alpha+2}}\sum_{1\leq k\leq m^2} \phi_{\sigma(k)} \cdot\Psi_\alpha\left[\sum_{t=1}^{k-1} \sqrt {\phi_{\sigma(t)}}\right] + \frac{C_\alpha}{m^4}\sum_{1\leq k\leq m^2} \phi_{\sigma(k)} \cdot \alpha\sqrt {\phi_{\sigma(k)}}\\
  &\leq C_\alpha\int_\cK g_\alpha (\phi) + \frac{C_\alpha}{m^2}\int_\cK \phi^{3/2}. 
\end{align*}
We can take  a finer subdivision and take $m$ sufficiently large so that finally,
\begin{equation*}
  \liminf_{n\to\infty} \frac{\mathbb E\left[ l_{\Psi-TRP}\right]}{n^{1+\alpha/2}} 
  \leq \tilde C_\alpha\int_\cK g_\alpha (\phi),
\end{equation*}
where $\tilde C_\alpha=2C_\alpha$. Note that with the same proof we can get the same result with $\tilde C_\alpha = \beta_{TSP}^\alpha.$ This ends the proof.
\end{proof}

\subsection{Technical lemma}

\begin{lemma}
\label{lemma:piece-wise constant approx psi-trp}
Let $f$ be a density on $\cK\subset [0,1]^2$. For any $\epsilon>0$, there exists a density $\phi$ of the form
\begin{equation*}
    \phi(x) = \sum_{k=1}^{m^2}\phi_k \mb{1}_{Q_k}(x)
\end{equation*}
such that
\begin{equation*}
    \|\phi-f\|_1 \leq \epsilon,\quad\text{and}\quad \left|\int_{\cK} g_\alpha(\phi)-\int_{\cK}g_\alpha(f) \right| \leq \epsilon.
\end{equation*}
\end{lemma}

\begin{proof}
Let $\delta>0$ an error parameter. Similarly to the proof of Lemma 7 from \citep{blanchard2021probabilistic}, we first take a density $\phi_\epsilon$ of the right form such that $\epsilon\leq \delta$ and
\begin{equation*}
    \|\phi_\epsilon-f\|_1 \leq \epsilon,\quad\text{and}\quad \|\sqrt{\phi_\epsilon}-\sqrt{f}\|_1\leq \epsilon.
\end{equation*}
We choose $\phi_\epsilon$ such that all $\phi_k$ are distinct. We will now write $\phi$ instead of $\phi_\epsilon$. Again, $\|\sqrt f\|_1,\|\sqrt\phi\|_1\leq 1$. We first introduce a new function $\tilde g_\phi$ which we will use as intermediary.
\begin{equation*}
    \tilde g_\alpha(\phi,x) = \begin{cases} 
          \phi(x)\left(\int_\cK\sqrt \phi \cdot \mb 1_{f< f(x)}\right)^\alpha & \text{if }\int_\cK \mb 1_{f=f(x)}=0 ,\\
          \sqrt{f(x)}\frac{\left(\int_\cK\sqrt{f}\cdot \mb 1_{f\leq f(x)}\right)^{\alpha+1}-\left(\int_\cK\sqrt{f}\cdot \mb 1_{f<f(x)}\right)^{\alpha+1}}{(\alpha+1)\int_\cK \mb 1_{f=f(x)}}  & \text{otherwise,}
       \end{cases}
\end{equation*}
Let us start by giving an estimate that will later be useful.
\begin{align*}
    \left|\left(\int_\cK\sqrt \phi \cdot \mb 1_{f< f(x)}\right)^\alpha-\left(\int_\cK\sqrt f \cdot \mb 1_{f< f(x)}\right)^\alpha\right|&\leq \int_\cK \left|\sqrt \phi-\sqrt f\right| \cdot \mb 1_{f< f(x)}\\
    &\cdot \alpha \max\left(\int_\cK\sqrt \phi \cdot \mb 1_{f< f(x)}, \int_\cK\sqrt f \cdot \mb 1_{f< f(x)}\right)^{\alpha-1}\\
    &\leq \alpha\cdot \delta
\end{align*}
The first step will be to compare $\tilde g_\alpha(\phi)$ and $g_\alpha(f)$. Similarly to the proof of the lower bound of Theorem 3 from \citep{blanchard2021probabilistic}, we can define the function
\begin{equation*}
    h(z):= \int_\cK \mb 1_{f=z}.
\end{equation*}
Recall that $h$ is non-zero only on a countable number of values which we will denote by $z_i$ for $i\geq 1$. Then,
\begin{align*}
    \left|\int_\cK \tilde g_\alpha(\phi) - g_\alpha(f)\right| &= \int_{f\notin\{z_1,\cdots,z_Z\}} \left|\phi(x)\left(\int_\cK\sqrt \phi \cdot \mb 1_{f< f(x)}\right)^\alpha-f(x)\left(\int_\cK\sqrt f \cdot \mb 1_{f< f(x)}\right)^\alpha\right|\\
    &\leq \int_{f\notin\{z_1,\cdots,z_Z\}} |\phi(x)-f(x)| \left\|\sqrt f \right\|_1^{\alpha/2} + \phi(x) \cdot\epsilon\alpha \max\left(\left\|\sqrt{f}\right\|_1^{\alpha-1},\left\|\sqrt{\phi}\right\|_1^{\alpha-1}\right)\\
    &\leq (1+\alpha)\delta.
\end{align*}
We now compare $g_\alpha(\phi)$ to $\phi(\cdot)\left(\int_\cK \sqrt \phi \cdot \mb 1_{\phi< \phi(\cdot)}\right)^\alpha$. For all $1\leq k\leq m^2$, we will denote $\eta_k:=\int_\cK \sqrt \phi \cdot \mb 1_{\phi< \phi_k}$. We now use the fact that all $\phi_k$ are distinct. By definition of $g_\alpha(\phi)$,
\begin{align*}
    \int_\cK \left|g_\alpha(\phi, x) - \phi(x)\left(\int_\cK \sqrt \phi \cdot \mb 1_{\phi< \phi(x)}\right)^\alpha \right|dx &= \sum_{k=1}^{m^2} \phi_k\int_{t=0}^{\mathcal A(Q_k)} \left[(\eta_k+t)^\alpha - \eta_k^\alpha\right] dt\\
    &\leq \sum_{k=1}^{m^2} \frac{\phi_k}{m^2} \cdot \frac{\alpha}{m^2} \left(\int_\cK\sqrt \phi\right)^{\alpha-1}\\
    &\leq \frac{\alpha}{m^2}.
\end{align*}
We take $m$ sufficiently large so that the left term can be upper bounded by $\delta.$
We now turn to comparing $\tilde g_\alpha(\phi)$ and $\phi(\cdot)\left(\int_\cK \sqrt \phi \cdot \mb 1_{\phi< \phi(\cdot)}\right)^\alpha$.
\begin{multline}
\label{eq:psi-trp complete inequality}
    \left|\int_\cK \left[\tilde g_\alpha(\phi,x) - \phi(x)\left(\int_\cK \sqrt \phi \cdot \mb 1_{\phi< \phi(x)}\right)^\alpha \right]dx\right|\\
    \leq \sum_{i\geq1} \left|\int_{f=z_i} \left[\tilde g_\alpha(\phi,x) - \phi(x)\left(\int_\cK \sqrt \phi \cdot \mb 1_{\phi< \phi(x)}\right)^\alpha \right]dx\right|\\
    + \left|\int_{f\notin\{z_i,\;i\geq 1 \}} \left[\phi(x)\left(\int_\cK\sqrt \phi \cdot \mb 1_{f< f(x)}\right)^\alpha  - \phi(x)\left(\int_\cK \sqrt \phi \cdot \mb 1_{\phi< \phi(x)}\right)^\alpha \right]dx\right|.
\end{multline}
Let us analyze the second term in the right-hand side of the inequality.
\begin{align*}
    \left|\int_{f\notin\{z_i,\;i\geq 1\}} \left[ \phi(x)  \left(\int_\cK\right.\right.\right. &\left.\left.\left.\sqrt \phi \cdot \mb 1_{f< f(x)}\right)^\alpha  - \phi(x)\left(\int_\cK \sqrt \phi \cdot \mb 1_{\phi< \phi(x)}\right)^\alpha \right]dx\right|\\
    &\leq \int_{f\notin\{z_i,\;i\geq 1\}} \phi(x)\cdot \alpha \left\|\sqrt\phi\right\|_1^{\alpha-1} \int_\cK \sqrt \phi\cdot \left|\mb 1_{f< f(x)}-\mb 1_{\phi< \phi(x)}\right| dx\\
    &\leq \alpha\iint_{\cK^2} \phi(x)\sqrt{\phi(y)}\left|\mb 1_{f(y)< f(x)}-\mb 1_{\phi(y)< \phi(x)}\right| dx dy\\
    &\leq\alpha  \iint_{\cK^2}\phi(x)\sqrt{\phi(y)} \mb 1_{|f(y)-f(x)|\leq 3\epsilon}\mb 1_{f(x)\neq f(y)}dx dy.
\end{align*}
By the dominated convergence theorem, this term vanishes as $\epsilon\to 0$. We take $0<\epsilon\leq \delta$ sufficiently small such that this term is upper bounded by $\delta$. We now turn to the first term of Equation~\eqref{eq:psi-trp complete inequality}. For $1\leq i\leq Z$, denote by $\eta_i:=\int_\cK \sqrt{f}\cdot \mb 1_{f<z_i}$. Then,
\begin{align*}
    &\sum_i\left|\int_{f(x)=z_i} \left[\tilde g_\alpha(\phi,x) - \phi(x)\left(\int_\cK \sqrt \phi \cdot \mb 1_{\phi< \phi(x)}\right)^\alpha \right]dx\right|\\
    &= \sum_i\left|\int_0^{\mathcal A(f=z_i)} z_i \left(\eta_i+\sqrt{z_i}t\right)^\alpha dt- \int_{f(x)=z_i}\phi(x) \left(\int_\cK \sqrt \phi \cdot \mb 1_{\phi< \phi(x)}\right)^\alpha \right|\\
    &\leq \sum_i z_i \left|\int_0^{h(z_i)} \left(\eta_i+\sqrt{z_i}t\right)^\alpha dt- \int_{f(x)=z_i} \left(\int_\cK \sqrt f \cdot \mb 1_{\phi< \phi(x)}\right)^\alpha \right| +\epsilon \cdot \left\|\sqrt f\right\|_1^{\alpha}\\
    &+ \int_\cK \phi(x)\cdot \alpha \max\left(\left\|\sqrt\phi\right\|_1^{\alpha-1},\left\|\sqrt f\right\|_1^{\alpha-1}\right) \epsilon dx\\
    &\leq \sum_i z_i \left|\int_0^{h(z_i)} \left(\eta_i+\sqrt{z_i}t\right)^\alpha dt- \int_{f(x)=z_i} \left(\int_\cK \sqrt f \cdot \mb 1_{\phi< \phi(x)}\right)^\alpha \right| + (1 + \alpha)\delta.
\end{align*}
For $x\in \cK$ such that $f(x)=z_i$,
\begin{align*}
    \int_\cK \sqrt f \cdot \mb 1_{\phi< \phi(x)} &=  \int_{|f-f(x)|> 3\epsilon} \sqrt f \cdot \mb 1_{f< f(x)} + \int_{|f-f(x)|\leq 3\epsilon} \sqrt f \cdot \mb 1_{\phi< \phi(x)}\\
    &= \int_\cK \sqrt f \cdot \mb 1_{f< z_i-3\epsilon} + \sqrt{z_i}\int_{f=z_i} \mb 1_{\phi< \phi(x)} + \int_\cK \sqrt f \cdot \mb 1_{\phi< \phi(x)}\mb 1_{|f-z_i|\leq 3\epsilon}\mb 1_{f\neq z_i}.
\end{align*}
Therefore,
\begin{equation*}
    \left|\int_\cK \sqrt f \cdot \mb 1_{\phi< \phi(x)} - \eta_i - \sqrt{z_i}\int_{f=z_i} \mb 1_{\phi< \phi(x)}\right| \leq 2 \int_\cK\sqrt f\cdot \mb 1_{|f-z_i|\leq 3\epsilon}\mb 1_{f\neq z_i}.
\end{equation*}
We can use this estimate for the following upper bound.
\begin{multline*}
    \sum_i z_i \left|\int_{f(x)=z_i}\left[\left(\eta_i+\sqrt{z_i}\int_{f=z_i}\mb 1_{\phi<\phi(x)}\right)^\alpha -\left(\int_\cK \sqrt f \cdot \mb 1_{\phi< \phi(x)}\right)^\alpha\right]dx \right|\\
    \leq 2\alpha
    \cdot \sum_i z_i \int_\cK\sqrt f\cdot \mb 1_{|f-z_i|\leq 3\epsilon}\mb 1_{f\neq z_i}.
\end{multline*}
All terms in the right-hand side sum vanish as $\epsilon\to 0$ by the dominated convergence theorem. Furthermore, the total sum is dominated by $2\alpha$. By monotone convergence, the sum vanishes as $\epsilon\to 0$. Let us take $\epsilon>0$ sufficiently small such that the left-hand term is upper bounded by $\delta$. The last term to analyze is
\begin{equation*}
    \sum_i z_i \int_{f(x)=z_i}\left(\eta_i+\sqrt{z_i}\int_{f=z_i}\mb 1_{\phi<\phi(x)}\right)^\alpha dx.
\end{equation*}
Let us take $Z\geq 1$ sufficiently large such that
\begin{equation*}
    \left|\sum_{i>Z} z_i h(z_i)\right|\leq \delta.
\end{equation*}
In particular, we can restrict the analysis to terms $1\leq i\leq Z$ since
\begin{equation*}
    \sum_{i>Z} z_i \left|\int_0^{h(z_i)} \left(\eta_i+\sqrt{z_i}t\right)^\alpha dt- \int_{f(x)=z_i} \left(\int_\cK \sqrt f \cdot \mb 1_{\phi< \phi(x)}\right)^\alpha \right|\leq 2 \sum_{i>Z} z_i h(z_i) \leq 2\delta .
\end{equation*}
Because $\{x:\;f(x)=z_i\}$ is measurable, for any tolerance $\epsilon>0,$ there exists $m_0\geq 1$ arbitrarily large and a set of sub-squares $E_i\subset\{1,\cdots, m^2\}$ such that
\begin{equation*}
    \left\|\mb 1_{f=z_i} - \mb 1_{\bigcup_{k\in E_i}Q
    _k}\right\|_1\leq \frac{\delta}{Z}.
\end{equation*}
for all $1\leq i\leq Z$. Then,
\begin{align*}
    &\sum_{i\leq Z}\left|\int_{f(x)=z_i}\left(\eta_i+\sqrt{z_i}\int_{f=z_i}\mb 1_{\phi<\phi(x)}\right)^\alpha dx - \int_{\bigcup_{k\in E_i}Q
    _k}\left(\eta_i+\sqrt{z_i}\int_{\bigcup_{k\in E_i}Q
    _k}\mb 1_{\phi<\phi(x)}\right)^\alpha dx\right| \\
    &\leq \sum_{i\leq Z}\int_{f(x)=z_i}\left|\left(\eta_i+\sqrt{z_i}\int_{f=z_i}\mb 1_{\phi<\phi(x)}\right)^\alpha - \left(\eta_i+\sqrt{z_i}\int_{\bigcup_{k\in E_i}Q
    _k}\mb 1_{\phi<\phi(x)}\right)^\alpha\right| dx\\
    &+ \sum_{i\leq Z}\int_\cK\left(\eta_i+\sqrt{z_i}\int_{\bigcup_{k\in E_i}Q
    _k}\mb 1_{\phi<\phi(x)}\right)^\alpha \left|\mb 1_{f(x)=z_i} - \mb 1_{\bigcup_{k\in E_i}Q
    _k}\right| dx\\
    &\leq \sum_{i\leq Z} h(z_i)\cdot \sqrt{z_i}\epsilon\cdot \alpha  + \sum_{i\leq Z} \frac{\delta}{Z}\\
    &\leq (1+\alpha)\delta.
\end{align*}
Now note that because values of $\phi$ on each sub-square $\phi_k$ are all distinct, then
\begin{equation*}
    \int_{\bigcup_{k\in E_i}Q
    _k}\left(\eta_i+\sqrt{z_i}\int_{\bigcup_{k\in E_i}Q
    _k}\mb 1_{\phi<\phi(x)}\right)^\alpha dx = \sum_{k=1}^{|E_i|} \frac{1}{m^2} \left(\eta_i+\sqrt{z_i}\frac{k-1}{m^2}\right)^\alpha.
\end{equation*}
Therefore,
\begin{align*}
    & \left| \int_0^{|E_i|/m^2} \left(\eta_i+\sqrt{z_i}t\right)^\alpha dt-\int_{\bigcup_{k\in E_i}Q
    _k}\left(\eta_i+\sqrt{z_i}\int_{\bigcup_{k\in E_i}Q
    _k}\mb 1_{\phi<\phi(x)}\right)^\alpha dx \right|\\
    &\leq \sum_{k=1}^{|E_i|} \frac{1}{m^2} \left[\left(\eta_i+\sqrt{z_i}\frac{k}{m^2}\right)^\alpha- \left(\eta_i+\sqrt{z_i}\frac{k-1}{m^2}\right)^\alpha\right]\\
    &\leq\frac{|E_i|}{m^2} \frac{\alpha}{m^2}(\eta_i+\sqrt{z_i})^{\alpha-1}\\
    &\leq \left(h(z_i)+\frac{\delta}{Z}\right)\frac{\alpha }{m^2}.
\end{align*}
We are now ready to merge all our estimates together.
\begin{align*}
     &\sum_{i\leq Z}\left|\int_0^{h(z_i)} \left(\eta_i+\sqrt{z_i}t\right)^\alpha dt- \int_{f(x)=z_i} \left(\int_\cK \sqrt f \cdot \mb 1_{\phi< \phi(x)}\right)^\alpha \right| \\
     & \leq \sum_{i\leq Z}\frac{\delta}{Z}
     +\sum_{i\leq Z}\left| \int_0^{|E_i|/m^2} \left(\eta_i+\sqrt{z_i}t\right)^\alpha dt-\int_{\bigcup_{k\in E_i}Q
    _k}\left(\eta_i+\sqrt{z_i}\int_{\bigcup_{k\in E_i}Q
    _k}\mb 1_{\phi<\phi(x)}\right)^\alpha dx \right|\\
    &+\sum_{i\leq Z}\left| \int_{\bigcup_{k\in E_i}Q
    _k}\left(\eta_i+\sqrt{z_i}\int_{\bigcup_{k\in E_i}Q
    _k}\mb 1_{\phi<\phi(x)}\right)^\alpha dx - \int_{f(x)=z_i}\left(\eta_i+\sqrt{z_i}\int_{f=z_i}\mb 1_{\phi<\phi(x)}\right)^\alpha dx\right|\\
    &+\sum_{i\leq Z}\left|\int_{f(x)=z_i}\left[\left(\eta_i+\sqrt{z_i}\int_{f=z_i}\mb 1_{\phi<\phi(x)}\right)^\alpha  - \left(\int_\cK \sqrt{f}\cdot \mb 1_{\phi<\phi(x)}\right)^\alpha\right]dx\right|\\
    &\leq (3+\alpha)\delta+ \frac{(1+\delta)\alpha}{m^2}.
\end{align*}
We can then take $m$ sufficiently large so that $\frac{(1+\delta)\alpha}{m^2}\leq \delta$. Finally, going back to Eq~\ref{eq:psi-trp complete inequality},
\begin{equation*}
    \left|\int_\cK \left[\tilde g_\alpha(\phi,x) - \phi(x)\left(\int_\cK \sqrt \phi \cdot \mb 1_{\phi< \phi(x)}\right)^\alpha \right]dx\right|\leq  (8+2\alpha) \delta
\end{equation*}
We now conclude by noting that
\begin{align*}
    \left|\int_\cK g_\alpha(\phi)-g_\alpha(f)\right|&\leq \left|\int_\cK \left[ g_\alpha(\phi,x)-\phi(x)\left(\int_\cK\sqrt{\phi}\cdot\mb 1_{\phi<\phi(x)}\right)^\alpha\right]  dx\right|\\
    &+\left|\int_\cK \left[\phi(x)\left(\int_\cK \sqrt \phi \cdot \mb 1_{\phi< \phi(x)}\right)^\alpha - \tilde g_\alpha(\phi,x) \right]dx\right|\\
    &+\left|\int_\cK \tilde g_\alpha(\phi) - g_\alpha(f)\right|\\
    &\leq  (10+3\alpha) \delta
\end{align*}
This is true for any $\delta>0$. This ends the proof of the lemma.
\end{proof}

\bibliographystyle{dinat}
\bibliography{references}

\end{document}